\documentclass[11pt,a4paper]{article}
\usepackage{graphicx}
\usepackage{amssymb,amsmath,amsthm}
\usepackage[english]{babel} 
\usepackage{t1enc}
\usepackage[latin2]{inputenc}
\usepackage{a4wide}

\makeatletter
\let\@@citation@@=\citation
\renewcommand{\citation}[1]{\@@citation@@{#1}%
	\@for\@tempa:=#1\do{\@ifundefined{cit@\@tempa}%
		{\global\@namedef{cit@\@tempa}{}}{}}%
}
\makeatother

\usepackage{hyperref}

\makeatletter
\def\@lbibitem[#1]#2#3\par{%
	\@ifundefined{cit@#2}{}{\@skiphyperreftrue
		\H@item[%
		\ifx\Hy@raisedlink\@empty
		\hyper@anchorstart{cite.#2\@extra@b@citeb}%
		\@BIBLABEL{#1}%
		\hyper@anchorend
		\else
		\Hy@raisedlink{%
			\hyper@anchorstart{cite.#2\@extra@b@citeb}\hyper@anchorend
		}%
		\@BIBLABEL{#1}%
		\fi
		\hfill
		]%
		\@skiphyperreffalse}%
	\if@filesw
	\begingroup
	\let\protect\noexpand
	\immediate\write\@auxout{%
		\string\bibcite{#2}{#1}%
	}%
	\endgroup
	\fi
	\ignorespaces
	\@ifundefined{cit@#2}{}{#3}}

\def\@bibitem#1#2\par{%
	\@ifundefined{cit@#1}{}{\@skiphyperreftrue\H@item\@skiphyperreffalse
		\Hy@raisedlink{%
			\hyper@anchorstart{cite.#1\@extra@b@citeb}\relax\hyper@anchorend
	}}%
	\if@filesw
	\begingroup
	\let\protect\noexpand
	\immediate\write\@auxout{%
		\string\bibcite{#1}{\the\value{\@listctr}}%
	}%
	\endgroup
	\fi
	\ignorespaces
	\@ifundefined{cit@#1}{}{#2}}
\makeatother

\mathcode`l="8000
\begingroup
\makeatletter
\lccode`\~=`\l
\DeclareMathSymbol{\lsb@l}{\mathalpha}{letters}{`l}
\lowercase{\gdef~{\ifnum\the\mathgroup=\m@ne \ell \else \lsb@l \fi}}%
\endgroup

\newtheorem{thm}{Theorem}
\newtheorem{cor}[thm]{Corollary}
\newtheorem{lem}[thm]{Lemma}
\newtheorem{prop}[thm]{Proposition}
\newtheorem{obs}[thm]{Observation}
\newtheorem{defi}[thm]{Definition}
\theoremstyle{definition}
\newtheorem*{remark*}{Remark}

\usepackage{indentfirst}
\usepackage{xspace} 

\def\F{\mbox{\ensuremath{\mathcal F}}\xspace}
\def\D{\mbox{\ensuremath{\mathcal D}}\xspace}

\begin{document}

\title{Aligned plane drawings of the generalized Delaunay-graphs for pseudo-disks}
%
%
\author{Bal\'azs Keszegh\thanks{Alfr\'ed R{\'e}nyi Institute of Mathematics and MTA-ELTE Lend\"ulet Combinatorial Geometry Research Group. Research supported by the National Research, Development and Innovation Office -- NKFIH under the grant K 116769 and by the Lend\"ulet program of the Hungarian Academy of Sciences (MTA), under grant number LP2017-19/2017.} \and
	D\"om\"ot\"or P\'alv\"olgyi\thanks{MTA-ELTE Lend\"ulet Combinatorial Geometry Research Group. Research supported by the Lend\"ulet program of the Hungarian Academy of Sciences (MTA), under grant number LP2017-19/2017.}}
%
%
%
\maketitle              

\begin{abstract}
	We study general Delaunay-graphs, which are natural generalizations of Delaunay triangulations to arbitrary families, in particular to pseudo-disks.
We prove that for any finite pseudo-disk family and point set, there is a plane drawing of their Delaunay-graph such that every edge lies inside every pseudo-disk that contains its endpoints. 

\end{abstract}

\section{Introduction}\setcounter{footnote}{0}

Delaunay triangulations play a central role in discrete and computational geometry. In many applications, however, one needs to deal with a different topology which requires to substitute disks in the definition with another family. If this other family consists of the homothets\footnote{A homothetic copy of a set is its scaled and translated copy (rotations are not allowed).} of some convex shape, then most properties generalize in a straight-forward manner \cite{Ma2000}. In this paper we study what happens when this is not the case, more precisely, we study the problem for families of (possibly non-convex) \emph{pseudo-disks}. Now we make the exact definitions.

\begin{defi}
Given a finite set of points $S$ and a family of regions $\F$, for a region $F\in \F$ we denote by $H_F=S\cap F$, the subset of $S$ defined by $F$. The vertices of the \emph{Delaunay-hypergraph} of $S$ with respect to $\F$ correspond to the points of $S$ and its hyperedges are $\{H_F:F\in\F\}$ (with multiplicities removed). The Delaunay-graph $\D({S,\F})$ is the graph formed by the size-$2$ hyperedges of this hypergraph.
\end{defi}

Thus, the vertices of the \emph{Delaunay-graph $\D({S,\F})$ of $S$ with respect to $\F$} correspond to the points of $S$, and two vertices $p,q\in S$ are connected by an edge if there is an $F\in \F$ such that $H_F=S\cap F=\{p,q\}$.
If $S\subset \mathbb R^2$ and $\F$ is the family of disks, this gives back the usual definition of Delaunay triangulations (when no four points from $S$ are on a circle).
It is well-known that this graph with respect to disks is planar, moreover, drawing its edges as straight-line segments we get a plane drawing in which the drawing of an edge $pq$ lies inside every disk containing both $p$ and $q$.
This is true also when the regions are the homothets of some convex region, or more generally, when $\F$ is a pseudo-disk arrangement containing only convex regions (as we will soon see).

\begin{defi}
	A Jordan region is a (simply connected) closed bounded region whose
	boundary is a closed simple Jordan curve.
	A family of Jordan regions is called a \emph{family of pseudo-disks} if the boundaries of every	pair of the regions intersect in at most two points.
\end{defi}

For an example of a pseudo-disk family see Figure \ref{fig:example}.

For points with respect to pseudo-disks, if we draw the edges in an arbitrary way inside one of their defining pseudo-disks (that is, the edge connecting points $p$ and $q$ is drawn inside an $F\in \F$ for which $H_F=\{p,q\}$), we get a drawing in which non-adjacent edges intersect an even number of times, using the following simple lemma (Lemma 1 in \cite{buzaglo}).

\begin{lem}\cite{buzaglo}\label{lem:disjointedges}
	Let $D_1$ and $D_2$ be two pseudo-disks in the plane. Let $x$ and $y$ be two
	points in $D_1\setminus D_2$. Let $a$ and $b$ be two points in $D_2\setminus D_1$. Let $e$ be any Jordan
	arc connecting $x$ and $y$ that is fully contained in $D_1$. Let $f$ be any Jordan arc connecting $a$ and $b$
	that is fully contained in $D_2$. Then $e$ and $f$ cross an even number of times.
\end{lem}

The Hanani-Tutte theorem then implies that the Delaunay-graph of points with respect to pseudo-disks is planar. 

If we additionally assume that the pseudo-disks in the family are all convex, then just like in the case of disks and homothets of a convex region, we can draw the edges as straight-line segments. As the regions are convex, the drawing of an edge $pq$ indeed lies inside every pseudo-disk containing both $p$ and $q$. Furthermore, two adjacent edges never intersect, while non-adjacent edges could intersect at most once, and by Lemma \ref{lem:disjointedges} an even number of times, thus they also never intersect. Thus, this is a plane drawing of the Delaunay-graph. 
To summarize, this proves the following.

\begin{thm}\label{thm:classic}
Given a pseudo-disk family $\F$ which contains only convex pseudo-disks and given a finite point set $S$, the straight-line drawing of the Delaunay-graph of $S$ with respect to $\F$ is a planar graph (and every drawn edge $pq$ lies inside every pseudo-disk containing both $p$ and $q$, by their convexity).
\end{thm}

This was shown already by Matou\v{s}ek et al.~\cite{MSW90} who actually defined pseudo-disks differently, and required them to be convex (for the special case of homothets of a given convex shape see also \cite{BCCS}). 

\begin{defi}
Given a pseudo-disk family $\cal F$, if a drawing of a graph on vertex set $S$ has the property that every drawn edge $pq$ lies inside each pseudo-disk that contains both $p$ and $q$ (and possibly other points of $S$ as well), then we say that the drawing of the graph is \emph{aligned} with $\F$. 	
\end{defi}

A drawing of a Delaunay-graph of $S$ with respect to $\F$ is trivially aligned with $\F$ whenever the regions of $\F$ are convex and the edges are drawn as straight-line segments. On the other hand, when the regions are not necessarily convex, then a straight-line drawing of the edges usually is not an aligned drawing; when a region is not convex, then a straight-line segment connecting two points inside it may not be fully contained in the region. 

The aim of this paper is to prove that we can also get a plane drawing aligned with $\F$ even when the pseudo-disks are not necessarily convex.

\begin{figure}[t]
	\centering
	\includegraphics[height=4cm]{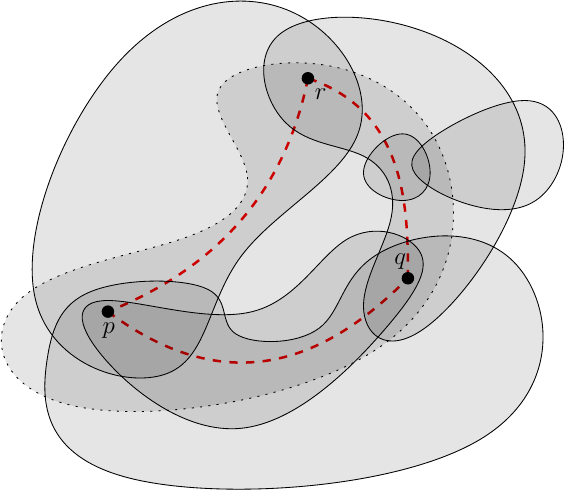}
	\caption{Aligned drawing of the Delaunay-graph of a $3$-element point set with respect to a pseudo-disk family.}
	\label{fig:example}	
\end{figure}

\begin{thm}\label{thm:planedelgraph}
	Suppose that we are given a finite pseudo-disk family $\F$ and a finite point set $S$, such that no point from $S$ is on the boundary of a pseudo-disk from $\F$.
	Then there is a plane drawing of the Delaunay-graph of $S$ with respect to $\F$ such that every edge $pq$ lies inside every pseudo-disk containing both $p$ and $q$.
\end{thm}

From now on we always assume (and maintain) that whenever a pseudo-disk family and a finite point set is given, then no point from the point set is on the boundary of a pseudo-disk from the family.

For an example for Theorem \ref{thm:planedelgraph}, see Figure \ref{fig:example}. Notice that in the example the edge connecting $p$ and $q$ must lie inside the two pseudo-disks that contain $p,q$ (but do not contain $r$), and also inside the pseudo-disk containing all three points (drawn with a dotted boundary in Figure \ref{fig:example}). 
Also notice that the edge connecting $r$ and $q$ cannot avoid going through a pseudo-disk which contains none of $r,q$. This shows that Theorem \ref{thm:planedelgraph} cannot be strengthened to additionally require that edges should be disjoint from pseudo-disks containing none of their endpoints (as is required, e.g., in \emph{clustered planarity} for clusters \cite{CB05}).

We expect that apart from the theoretical and esthetical interest, 
Theorem \ref{thm:planedelgraph} can be useful in several applications.
In fact, our motivation to study the problem came from the fact that this was exactly the lemma we needed in a recent joint result with Ackerman \cite{tuple} about certain colorings of the edges of the Delaunay-graph.

A quite similar problem have been studied by Kratochv\'il and Ueckerdt \cite{KU13}.
The main difference between our paper and \cite{KU13} is that the (not necessarily two-element) point sets that they are trying to connect inside the pseudo-disks are required to be disjoint, i.e., for the point set $P_i$ inside disk $D_i$ and for the point set $P_j$ inside disk $D_j$, they have $P_i\cap P_j=\emptyset$ for every $i$ and $j$.
They prove that under these conditions such pairwise non-crossing connecting curves $P_i\subset\gamma_i\subset D_i$ always exist.
At the end of their paper they propose to relax the condition $P_i\cap P_j=\emptyset$, and claim that then non-crossing connecting curves might not exist; this is because in their version points from $P_i$ can be arbitrarily contained in other $D_j$.
We mention that there are other papers that study drawing edges inside certain restricted regions, e.g., Silveira, Speckmann and Verbeek \cite{SSV17}.

We also note that this paper does not aim to present the most concise proof of our main result, instead it tries to be as self-contained as possible, exposing lemmas which we find useful along the way. Additionally, we made sure not to use implicitly any extra conditions about how nice the drawings of the pseudo-disks are (as it happens sometimes in the literature).


\section{Proof of Theorem \ref{thm:planedelgraph}}\label{sec:proof}

We will need the following lemma:

\begin{lem}\label{lem:ray}
	Given a finite family of pseudo-disks such that all pseudo-disks contain a common point $p$, any point $q$ can be connected by a Jordan curve to $p$ such that this curve does not intersect the boundary of pseudo-disks containing $q$ (besides $p$) and intersects exactly once the boundary of pseudo-disks that do not contain $q$ (but contain $p$).
\end{lem}

A family of pseudo-disks is in \emph{general position} if there are no three pseudo-disks whose boundaries intersect in a common point.

In \cite{lenses} it is proved that for a finite family of pseudo-disks in general position and all containing a common point $p$ there exists a combinatorially equivalent family of pseudo-disks, all of which are star-shaped\footnote{A region is star-shaped with respect to $p$ if every line through $p$ intersects the region in a segment containing $p$.} with respect to $p$. This clearly implies the above lemma when the pseudo-disk family is in general position. We will see later that actually this is enough for us whenever we will apply this lemma. Alternately, appropriate application of the Sweeping theorem of Snoeyink and Hershberger \cite{SH89} also implies Lemma \ref{lem:ray} without assuming general position, yet the proof of their theorem is quite involved. Finally, in \cite{anchored,anchoredArxiv} a relatively simple self-contained proof of Lemma \ref{lem:ray} is shown, again without assuming general position. 

We need some further definitions and lemmas before we can prove our main result.  From now on every point set $S$ we consider is finite, even if we do not always emphasize this.

\begin{defi}
	Given two pseudo-disks whose boundaries intersect (that is, two Jordan regions whose boundaries intersect twice), after removing their boundaries the plane is split into three bounded and one unbounded region\footnote{This is implied by the Jordan curve theorem.}. We call the bounded regions the three \emph{lenses} defined by these two pseudo-disks. If a set of points $S$ is given, we say that \emph{a lens is empty} if it does not contain any point from $S$.
\end{defi}

\begin{defi}
	In a pseudo-disk family $\F$, replacing a pseudo-disk $F$ with some $F'\subsetneq F$ such that the new family is still a pseudo-disk family, is called a \emph{(geometric)\footnote{The word `geometric' in the definition (which we will omit in the rest of the paper) emphasizes here the difference from a similar definition (e.g., \cite{tuple}) where only $F'\cap S\subset F\cap S$ is required instead of $F'\subset F$; see also Section \ref{sec:planarsupports}.} shrinking} of $F$ to $F'$. If such an $F'$ already exists in $\F\setminus \{F\}$, then simply deleting $F$ from $\F$ is also called a shrinking of $F$ to $F'$. Given a point set $S$, such a shrinking is \emph{hypergraph preserving} on $S$ if $F'\cap S=F\cap S$. 
	
	Applying shrinking steps to multiple members of the family $\F$ after each other is called a \emph{shrinking} of $\F$.
	A shrinking of $\F$ is \emph{hypergraph preserving} if all shrinking steps are hypergraph preserving.
\end{defi}

\begin{obs}
	If we do a hypergraph preserving shrinking on $\F$ to get $\F'$, then by definition in each step if $F\in \F$ is shrunk to $F'\in \F'$, then we have $H_F=H_{F'}$ and thus the Delaunay-hypergraph of $S$ with respect to $\F$ is the same as that of $S$ with respect to $\F'$. That is, as its name suggests, a hypergraph preserving shrinking does indeed preserve the Delaunay-hypergraph of $S$ with respect to $\F$.
\end{obs}

\begin{lem}\label{lem:emptylens}
	Given a point set $S$ and a finite family $\F$ of pseudo-disks, suppose that there is a containment-minimal empty lens $L$ defined by a pair of pseudo-disks $F_1,F_2\in \F$ with $L\subset F_1$, then we can apply a hypergraph preserving shrinking on $\F$ to obtain a pseudo-disk family which has strictly fewer number of intersections of the boundaries.
\end{lem}

\begin{proof}	
	The main idea of the proof is ``to get rid of'' $L$. Such a removal of an empty lens is quite standard and the reader may want to skip the rest of the proof. As the aim of this paper is to give tools that can be used safely for any family of pseudo-disks, we decided to include a proof that takes care of topologically non-intuitive cases as well.
	
	We prove that we can shrink $F_1$ to some $F_1'$ such that $\F'=\F\setminus \{F_1\}\cup \{F_1'\}$ is a pseudo-disk family, $F_1'\cap S=F_1\cap S$ (that is, shrinking $F_1$ to $F_1'$ is hypergraph preserving) and $F_1'\cap L=\emptyset$ (that is, we got rid of $L$).  

	Let $l_1$ (resp.\ $l_2$) be the maximal curve that is on the boundary of both $F_1$ (resp.\ $F_2$) and $L$. There might be some (empty) pseudo-disks that lie completely inside $L$. We claim that every maximal curve inside $L$ which is part of a boundary of some pseudo-disk different from $F_1$ and $F_2$ and is not completely inside $L$, has one endpoint on $l_1$ and another on $l_2$. Indeed, if such a maximal curve on the boundary of some $F_3$ would have both endpoints on $l_1$ (resp.\ $l_2$), then $F_1$ (resp.\ $F_2$) and $F_3$ would define a lens which lies inside $L$ contradicting its containment minimality.
	
\begin{figure}[t]
	\centering
	\includegraphics[height=4cm]{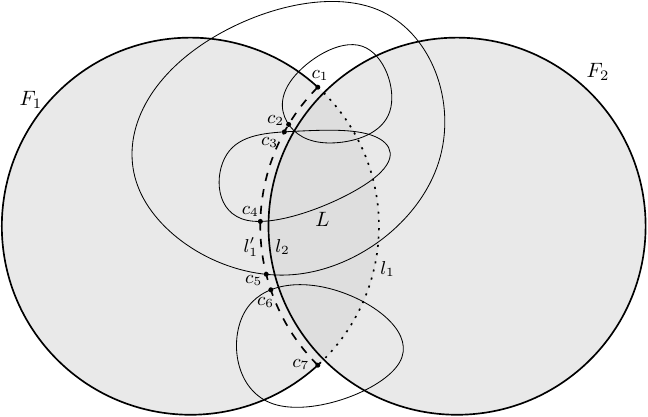}
	\caption{Removing the lens $L=F_1\cap F_2$ from $F_1$.}
	\label{fig:lensremoval}	
\end{figure}

\begin{figure}[t]
	\centering
	\includegraphics[height=4cm]{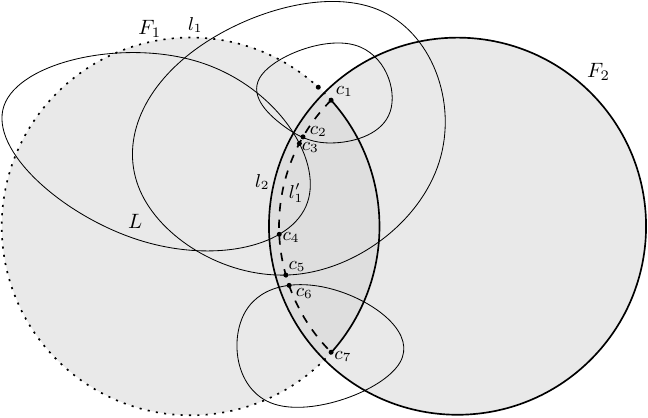}
	\caption{Removing the lens $L=F_1\setminus F_2$ from $F_1$.}
	\label{fig:lensremovalb}	
\end{figure}
	
	Now we are ready to shrink $F_1$. Basically we want to delete $L$ from $F_1$, but we have to shrink it a bit afterwards to avoid the introduction of common boundary parts. This turns out to be a bit technical;\footnote{It would be convenient to delete from $F_1$ an $\epsilon$-expansion of $L$ but this would not always work. For example it is possible that the boundary of another pseudo-disk makes infinitely many smaller and smaller ``squiggles'' (like $x\sin\frac 1x$) before intersecting $L$ and thus intersects every $\epsilon$-expansion of $L$ (where $\epsilon<\epsilon_0$ for some small $\epsilon_0$) more than twice.} there are two cases (when $L=F_1\cap F_2$ and when $L=F_1\setminus F_2$), yet luckily we can handle both of them exactly the same way (see Figure \ref{fig:lensremoval} and \ref{fig:lensremovalb}). In order to do that we will define a curve inside $F_1\setminus L$ which intersects the boundary of every pseudo-disk the same number of times as $l_2$ does. Then we shrink $F_1\setminus L$ by essentially replacing $l_2$ by this curve on the boundary. Next we give the details of how we do all of this.	
	
	Given a finite family of Jordan regions whose boundaries intersect finite many times, the \emph{vertices} of the \emph{arrangement determined by these regions} are the intersection points of the boundaries of the regions, the \emph{edges} are the maximal connected parts of the boundaries of the regions that do not contain a vertex, and the \emph{cells} are the maximal connected parts of the plane which are disjoint from the edges and the vertices of the arrangement.

	We take the arrangement determined by the boundaries of the given family of pseudo-disks.	
	Consider the vertices (that is, intersection points) of the arrangement that are on $l_2$. For every such vertex and every edge incident to this vertex and lying inside $F_1\setminus L$, but not on $l_2$, we take a small part of that edge ending in this vertex and call this a half-edge. These half-edges can be ordered naturally first according to the order of their endpoints on $l_2$, second for two half-edges sharing an endpoint we order them according to their rotation order around this vertex. For every consecutive pair of half-edges in this ordering there is also a unique cell of the arrangement that is `between' these half-edges in this ordering.
	
	Thus we get a natural ordering of the half-edges, $e_1,e_2,\dots e_s$. Notice that the first and last half-edges lie on the boundary of $F_1\setminus L$. Now for every $e_i$ choose an arbitrary point $c_i$ on it and connect for every $1\le i\le s-1$ the points $c_i$ and $c_{i+1}$ by a curve lying inside the cell that lies between them. While we need to draw several curves inside one cell (e.g., the curves connecting $c_1$ with $c_2$ and $c_4$ with $c_5$ in Figure \ref{fig:lensremoval}), we can draw these curves such that no two of them intersect, as we can draw the curve connecting $c_i$ and $c_{i+1}$ close to their half-edges and (if they do not share an endpoint) the part of $l_2$ separating their endpoints. 
	The union of all the curves is a curve $l_1'$ that connects $c_1$ to $c_s$.
	
	Let $F_1'$ be the region whose boundary consists of $l_1'$ and the boundary part of $F_1$ from $c_1$ to $c_s$ which is disjoint from $l_2$.
	Clearly we can draw the curves forming $l_1'$ such that at the end $F_1'\cap S=F_1\cap S$.

	Having defined $F_1'$, the shrinking of $F_1$, we are left to prove that it has the properties we required.
	
	Clearly, $F_1'\cap L=\emptyset$ and we made sure that $F_1'\cap S=F_1\cap S$. So we need to show only that the new family is also a pseudo-disk family and has strictly fewer intersections between the boundaries of its members.
	
	Consider now an intersection of $l_1'$ with the boundary of some $F_3$. By definition, it must be $c_i$ for some $2\le i\le s-1$. The half-edge $\gamma_i$ containing $c_i$ has an endpoint on $l_2$, which is an intersection point of the boundaries of $F_3$ and $F_2$. Using our observation from the beginning of the proof, the maximal curve inside $L$ whose starting point is this intersection has the other endpoint on $l_1$ and is an intersection point of the boundary of $F_1$ with the boundary of $F_3$. Arguing now in the opposite direction, for every intersection point of $l_1$ with a boundary of some $F_3$ there is a corresponding intersection point on $l_2$ and then also on $l_1'$ with the boundary of $F_3$. We conclude that the intersection points on the boundary of $F_1$ and $F_1'$ are in bijection except for the two intersection points of the boundaries of $F_1$ and $F_2$ as the boundaries of $F_1'$ and $F_2$ do not intersect. This implies that the family is still a pseudo-disk family and that the overall number of intersection points of boundaries decreased by $2$, finishing the proof.	
\end{proof}

\begin{defi}\label{def:respects}
	Given a point set $S$, we say that a pseudo-disk family  $\F$ \emph{respects} $S$ if for every pair of pseudo-disks $F_1,F_2\in \F$, 
	\begin{equation} \label{eq:respect1}
	\text{if } {F_1\cap S}\subseteq {F_2\cap S,} \text{ then } F_1\subseteq F_2 \text{ as well and}
	\end{equation}	
	\begin{equation} \label{eq:respect2}
	\text{if } {(F_1\cap S)\cap (F_2\cap S)}=\emptyset, 
	\text{ then } F_1\cap F_2=\emptyset \text{ as well.}
	\end{equation}
\end{defi}

\begin{obs}\label{obs:respectunique}
	If a pseudo-disk family $\F$ respects $S$, then by definition for every subset $S'\subseteq S$ there is at most one pseudo-disk $F$ such that $F\cap S=S'$.
\end{obs}

Indeed, if $F_1\cap S=F_2\cap S$, then by assumption (\ref{eq:respect1}) of Definition \ref{def:respects} both $F_1\subseteq F_2$ and $F_2\subseteq F_1$, that is $F_1=F_2$.
\begin{lem}\label{lem:respect}
	Given a point set $S$ and a finite pseudo-disk family $\F$, we can shrink $\F$ to get a finite pseudo-disk family $\F'$ such that 
	\begin{itemize}
		\item [(i)]
			this shrinking is hypergraph preserving on $S$,
		\item [(ii)]
			$\F'$ respects $S$ and
		\item [(iii)]
		 	no pair of pseudo-disks in $\F'$ defines an empty lens.
	\end{itemize}
	
\end{lem}

\begin{proof}
We keep applying Lemma \ref{lem:emptylens} to a containment-minimal empty lens until there are no more empty lenses. This is a finite process as in each step the number of intersections between boundaries of pseudo-disks decreases. By Lemma \ref{lem:emptylens} it follows that the new family is a pseudo-disk family and that this shrinking was hypergraph preserving on $S$.

Next, if there is a pair of pseudo-disks which intersect $S$ in the same subset $S'$, then since there are no empty lenses, one of them must be contained in the other.
The bigger one can be shrunk to the smaller, so we can delete it.
We keep doing this until for every $S'$ there is only at most one pseudo-disk that intersects $S$ in $S'$.

Finally, it is easy to see that if there was a pair of pseudo-disks in $\F'$ for which (\ref{eq:respect1}) or (\ref{eq:respect2}) from  Definition \ref{def:respects} did not hold, then either they would intersect $S$ in the same subset or one of the lenses they form would be empty, contradicting the fact that there were no more empty lenses in $\F'$.
\end{proof}

\begin{defi}
	Given a pseudo-disk family $\F$, the \emph{depth of a point} is the number of pseudo-disks containing this point. For a point with depth $d$ for some integer $d$ we say that it is \emph{$d$-deep}.
\end{defi}

\begin{remark*}
Notice that using Lemma \ref{lem:respect} from a pseudo-disk family $\F$ we can get another pseudo-disk family $\F'$ that defines the same hypergraph and in which drawing every Delaunay-edge arbitrarily inside its defining pseudo-disk (unique in $\F'$), disjoint Delaunay-edges are drawn without intersections, and this drawing is aligned with $\F'$ and also with $\F$. Thus, we are only left to deal with intersections between Delaunay-edges that share an edge. A possible solution, suggested by an anonymous reviewer, would be to split in some way the vertices (and shrinking the pseudo-disks appropriately) so that adjacent edges also become disjoint edges (stars became matchings in the Delaunay-graph) and apply Lemma \ref{lem:respect} to this to get a planar drawing, after which we somehow merge back the appropriate vertices (changing the drawing appropriately). Instead, we choose another (albeit possibly longer) route which is based on the forthcoming Lemma \ref{lem:transversaledge}, which strengthens Lemma \ref{lem:ray} in a specific setting (and indeed its proof uses Lemma \ref{lem:ray}) and which may be interesting in its own as well.
\end{remark*}

Before we present our second central lemma in proving Theorem \ref{thm:planedelgraph}, let us briefly discuss whether it can be assumed that the pseudo-disks are in general position.

It is easy to see that whenever in a pseudo-disk family $\F$ there are at least three pseudo-disk boundaries intersecting in a common point $p$, we can  get rid of this multiple intersection (and thus reduce the number of such intersections) by perturbing in an appropriately small disk around $p$ the pseudo-disk boundaries going through $p$ such that this perturbation is a shrinking, moreover if an $S$ is given, we can do this perturbation such that it is a hypergraph preserving shrinking. In fact this can be done similarly as we have removed an empty lens in the proof of Lemma \ref{lem:emptylens}, except that now $p$ plays the role of the empty lens.
Repeating such shrinking steps we can shrink $\F$ to another pseudo-disk family $\F'$ in which there are no three such pseudo-disks, that is, $\F'$ is in general position. If Theorem \ref{thm:planedelgraph} holds for $\F'$, then the same drawing implies that it also holds for $\F$. Thus, in the rest of the paper we can assume that the pseudo-disk families we deal with are in general position. In particular, whenever we apply Lemma \ref{lem:ray} we can just use it for a family in general position.

\begin{lem}\label{lem:transversaledge}
	We are given a point set $S$ and a family of pseudo-disks $\F$ that respects $S$ such that every pseudo-disk contains exactly two points from $S$. Given a pseudo-disk $F_{x,y}\in \F$ containing only the two points $x,y$ from $S$, we can draw a curve connecting $x$ and $y$ that lies completely inside $F_{x,y}$ and intersects the boundary of every other pseudo-disk at most once.
\end{lem}

Before starting the proof of Lemma \ref{lem:transversaledge}, let us emphasize that it is necessary to assume that $\F$ respects $S$, otherwise the lemma is not true. Indeed, as we have already noted earlier, on Figure \ref{fig:example} if we connect $q$ and $r$ inside all the pseudo-disks that contain $q$ and $r$, then we must intersect the boundary of some pseudo-disk twice. While in this example not all pseudo-disks contain exactly two points, one can easily draw an example where this holds as well.

We remark that if in Lemma \ref{lem:transversaledge} we additionally assume that there are no empty lenses defined by the pseudo-disks of $\F$ then the boundaries of the pseudo-disks inside $F_{x,y}$ behave like pseudo-lines and thus in this case we can easily apply Levi's Enlargement Lemma \cite{Levi} (for a recent proof of Levi's Enlargement Lemma see \cite{levinew}) to conclude the statement of the lemma. Moreover, this additional assumption could be indeed assumed when Lemma \ref{lem:transversaledge} will be used during the proof of Theorem \ref{thm:planedelgraph} as before applying the lemma we will apply Lemma \ref{lem:emptylens}, and so we can assume that there are no empty lenses.
Yet we have decided to state and prove this slightly more general statement, in order to be as self-contained as possible.

\begin{proof}[Proof of Lemma \ref{lem:transversaledge}]
	For any pair $p,q$ of points of $S$, denote by $F_{p,q}$ the unique pseudo-disk containing exactly these two points from $S$, if it exists (uniqueness follows from Observation \ref{obs:respectunique} as $\F$ respects $S$).
	
		We will draw the curve connecting $x$ and $y$ inside $F_{x,y}$.
	Thus, it cannot intersect any pseudo-disk that lies outside $F_{x,y}$, that is, as $\F$ respects $S$, any pseudo-disk that contains two points of $S$, both different from $x$ and $y$. Thus when drawing the arc, we only need to care about pseudo-disks of two \emph{types}, of type $F_{x,*}$, namely, a pseudo-disk $F_{x,p}$ for some $p$ different from $y$ and of type $F_{y,*}$, that is, a pseudo-disk $F_{y,q}$ for some $q$ different from $x$. Note that $F_{x,y}$ is not of any of these two types.
	
	If $p\ne q$, then $F_{x,p}\cap F_{y,q}=\emptyset$ as $\F$ respects $S$. This implies the following:
	
	\begin{prop}\label{prop:thmprop}
	If a point (not necessarily from $S$) inside $F_{x,y}$ is contained in at least $3$ pseudo-disks different from $F_{x,y}$ (that is, has depth at least $4$), then all these pseudo-disks must be of the same type.
	\end{prop}

\begin{figure}[t]
	\centering
	\includegraphics[height=7cm]{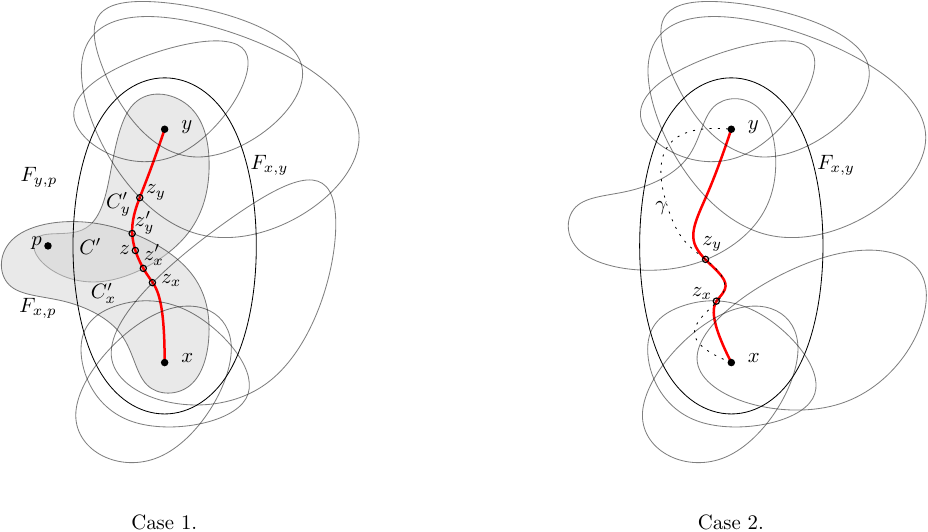}
	\caption{Drawing the curve connecting $x$ and $y$.}
	\label{fig:case12}	
\end{figure}

	Now we can continue with the proof of Lemma \ref{lem:transversaledge} --- for illustrations see Figure \ref{fig:case12}.
	
	\begin{itemize}
		\item Case 1. There is a point $z$ (not from $S$) in $F_{x,y}$ contained in pseudo-disks of both types. 
	
	Then by Proposition \ref{prop:thmprop} it must be a $3$-deep point which besides $F_{x,y}$ is contained only in $F_{x,p}$ and $F_{y,p}$ for some $p$. Also, in the arrangement of the pseudo-disks the cell containing $z$ must be bounded only by boundary parts of $F_{x,y},F_{x,p}$ and $F_{y,p}$, otherwise a point in a neighboring cell would contradict Proposition \ref{prop:thmprop}. In fact, it is easy to see that it must be bounded by parts of the boundaries of all of these three pseudo-disks, otherwise we would have an empty lens.
	
	Take now the arrangement defined by $F_{x,y}$ and the pseudo-disks of type $F_{x,*}$. Let $C_x$ be the cell containing $z$ in this arrangement. Take the first intersection point $z_x$ of the ray guaranteed by Lemma \ref{lem:ray} going from $z$ to $x$ with the boundary of $C_x$. Note that $C_x$ is disjoint from all pseudo-disks of type $F_{y,*}$ except for $F_{y,p}$. Moreover, again by Proposition \ref{prop:thmprop}, the boundary of $F_{y,p}$ intersects the boundary of $F_{x,y}$ and the boundary of $F_{x,p}$ in two points but does not intersect the boundaries of other pseudo-disks of type $F_{x,*}$. Thus, $F_{y,p}$ subdivides $C_x$ into at most three parts: $C'$ containing $z$, $C_x'$ having $z_x$ on its boundary, and to a possible third cell. It is easy to see that $C'$ and $C_x'$ must share boundary parts, and so we can choose a point $z_x'$ on their common boundary. Note that $C'$ is actually the cell containing $z$ in the arrangement of all pseudo-disks.
	
	Now take the arrangement defined by $F_{x,y}$ and the pseudo-disks of type $F_{y,*}$. A symmetric argument gives the cells $C'$ and $C_y'$ and the points $z_y$ and $z_y'$ (note that we get the same $C'$ as it is again the cell containing $z$ in the arrangement of all pseudo-disks).
	
	Now the curve connecting $x$ and $y$, and intersecting the boundary of every pseudo-disk at most once consists of the following
	parts: a curve from $x$ to $z_x$ along a ray guaranteed by Lemma \ref{lem:ray} (applied for the pseudo-disks containing $x$), a curve inside $C_x'$ from $z_x$ to $z_x'$, a curve from $z_x'$ to $z_y'$ inside $C'$ (which can go through $z$ if we wish to), a curve from $z_y'$ to $z_y$ inside $C_y'$ and finally a curve from $z_y$ to $y$ along a ray guaranteed by Lemma \ref{lem:ray} (applied for the pseudo-disks containing $y$). By Lemma \ref{lem:ray} every point of the curve connecting $x$ and $z_x$ is inside at least two pseudo-disks of type $F_{x,*}$ and thus by Proposition \ref{prop:thmprop} it cannot intersect any pseudo-disks of type $F_{y,*}$. Similarly, the curve connecting $y$ and $z_y$ cannot intersect any pseudo-disks of type $F_{x,*}$.
	Altogether, using again Lemma \ref{lem:ray}, we get that the union of these curves defines a curve intersecting every pseudo-disk boundary at most once, as required.
	
\item	Case 2. Every point in $F_{x,y}$ is contained only in pseudo-disks of one type.
	
	In this case going along an arbitrary curve $\gamma$ from $x$ to $y$, the last point $z_x$ which is contained in a pseudo-disk of type $F_{x,*}$ must be at least $2$-deep and contained in $F_{x,y}$ and in pseudo-disks only of type $F_{x,*}$ (if there are no pseudo-disks of type $F_{x,*}$, let $z_x=x$.). Going further along this curve towards $y$, there are $1$-deep points and then the first at least $2$-deep point $z_y$ must be contained in $F_{x,y}$ and in pseudo-disks only of type $F_{y,*}$ (or if there are none, let $z_y=y$.). Now we can apply Lemma \ref{lem:ray} to get a curve from $x$ to $z_x$  (applied for the pseudo-disks containing $x$). This will be disjoint from each $F_{y,*}$, as there are no points in $F_{x,y}$ contained in pseudo-disks of both types. We similarly connect $y$ to $z_y$ using Lemma \ref{lem:ray}. Together with the part of $\gamma$ connecting $z_x$ to $z_y$ we get a curve, which using again Lemma \ref{lem:ray} intersects every pseudo-disk boundary at most once, as required.		
	\end{itemize}
	\end{proof}

\begin{proof}[Proof of Theorem \ref{thm:planedelgraph}]
	Given a finite pseudo-disk family $\F$ and a point set $S$, we want to find a plane drawing of the Delaunay-graph of $S$ with respect to $\F$ such that every edge $pq$ lies inside every pseudo-disk containing both $p$ and $q$.
	
	First we shrink $\F$ using Lemma \ref{lem:respect} to get a family that respects $S$. As we did a hypergraph preserving shrinking on $S$, the new family has the same (possibly empty) Delaunay-graph as $\F$. Next we remove all the pseudo-disks containing at most one or at least three points from $S$, by which the Delaunay-graph is again left intact. We get the pseudo-disk family $\F'$. Notice that as $\F'$ respects $S$, for every pair of points $p,q$ that are connected by an edge in the Delaunay-graph, in $\F'$ there is exactly one pseudo-disk $F'_{pq}$ that contains these two points (and no other point of $S$).
	
	We claim that a plane drawing of the Delaunay-graph of $\F'$ with respect to $S$ such that for each edge $pq$ its drawing lies inside $F'_{pq}$ is a plane drawing of the Delaunay-graph of $\F$ with respect to $S$ as required. Indeed, take an arbitrary edge $pq$ of the Delaunay-graph of $\F$ with respect to $S$ and let $F$ be an arbitrary pseudo-disk such that $p,q\in F\in \F$. After shrinking $\F$ to $\F'$, $F$ was shrunk to some $F'$ (possibly in multiple steps) which must contain $F'_{pq}$ as $F'\cap S=F\cap S\supseteq \{p,q\}=F'_{pq}\cap S$ and $\F'$ respects $S$. Thus, the drawing of the edge $pq$ lies inside $F'_{pq}\subset F'\subset F$, as claimed.\footnote{We note that for this argument to work it was important that we removed the pseudo-disks containing at least $3$ points only after we shrank $\F$ to a family that respects $S$.}
	
	Thus, we are left to prove that there exists a plane drawing of the Delaunay-graph of $\F'$ with respect to $S$ such that for each edge $pq$ its drawing lies inside $F'_{pq}$. We will prove this by drawing the edges one-by-one using Lemma \ref{lem:transversaledge}.
	
	We additionally require that every drawn edge intersects the boundary of every pseudo-disk of $\F'$ at most once (which implies that it intersects a boundary only when it is necessary, that is, when exactly one of its endpoints is inside this pseudo-disk).\footnote{Note that we do not guarantee this additional property for the original pseudo-disks of $\F$. See also the remark before the proof of Lemma \ref{lem:transversaledge}.}
	We take the edges one-by-one in an arbitrary order. We draw the first edge using Lemma \ref{lem:transversaledge}. The additional requirement holds for this first drawn edge by Lemma \ref{lem:transversaledge}.
	
	Suppose that the next edge we want to draw connects $x$ and $y$ and we want to draw it inside $F_{x,y}$. For an illustration of the rest of the proof see Figure \ref{fig:addingedges}.
	Draw it first using Lemma \ref{lem:transversaledge}, then this drawn curve $f$ intersects the boundary of any other pseudo-disk at most once. Even though $f$ may intersect previously drawn edges, it can only intersect edges connecting $x$ or $y$ to some other points of $S$ as all other edges lie outside $F_{x,y}$. Indeed, such an edge, connecting points $p,q$ where $\{p,q\}\cap \{x,y\}=\emptyset$, is drawn inside $F_{p,q}$, which in turn is disjoint from $F_{x,y}$ using that $\F'$ respects $S$. Take the intersection $z_x$ of $f$ with a drawing of an edge $xp$ which is farthest from $x$ along $f$ among edges of this type (or $z_x=x$ if there is no such intersection). Take also the intersection $z_y$ of $f$ with a drawing of an edge $yq$ which is farther from $x$ than $z_x$ along $f$ and is the closest to $x$ among edges of this type (or $z_y=y$ if there is no such intersection). Now take the curve $f'$ which goes from $x$ to $z_x$ along the already drawn $xp$ (very close to and on the appropriate side - as all boundaries are Jordan curves intersecting at most twice, this can be done such that a part of a boundary curve and a curve running close to it intersects exactly the same boundary curves) on which $z_x$ lies, then goes along $f$ from $z_x$ to $z_y$ and then goes from $z_y$ to $y$ along the already drawn $yq$ (again very close to and on the appropriate side) on which $z_y$ lies. By appropriate sides we mean that we choose sides such that $f'$ does not intersect the edges $xp$ and $yq$ apart from their endpoints. This gives a drawing $f'$ of the edge $xy$ which does not intersect any earlier edge. 
	
	We need to prove that $f'$ lies inside $F_{x,y}$. As $f$ lies inside $F_{x,y}$, we only need to care about the two parts that are drawn along the drawings of the edges $xp$ and $yq$. Yet by induction these edges were drawn such that they intersect $\partial F_{x,y}$ once which implies that along these edges the points inside $F_{x,y}$ form a connected component. Thus the edge parts that connect $x$ to $z_x$ and $y$ to $z_y$ lie inside $F_{x,y}$ as $x,z_x,y,z_y$ are all inside $F_{x,y}$. That is, all three parts of $f'$ lie inside $F_{x,y}$, as required.
	
	\begin{figure}[t]
		\centering
		\includegraphics[height=4cm]{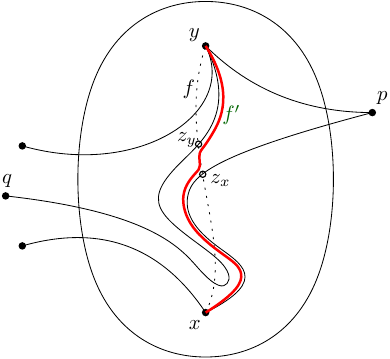}
		\caption{Adding the drawing of the edge $xy$.}
		\label{fig:addingedges}	
	\end{figure}

	Now we are left to prove that $f'$ intersects the boundary of every pseudo-disk at most once.
	For this we prove that for an arbitrary pseudo-disk $F$ its boundary $\partial F$ intersects $f'$ the same number of times as it intersects $f$. Denote by $f'_x$ the part of $f'$ between $x$ and $z_x$ and by $f_x$ the part of $f$ between $x$ and $z_x$.
	Both $f_x$ and $f'_x$ can intersect $\partial F$ at most once as $f_x$ is part of $f$ while $f'_x$ is part of the drawn edge $xp$, both of which intersect $\partial F$ at most once. This implies that $f'_x$ intersects $\partial F$ if and only if $f_x$ does as both of these happens if and only if $\partial F$ separates $x$ from $z_x$.
	The similar statement holds for the part of $f'$ between $y$ and $z_y$.
	
	As the remaining (middle) parts of $f$ and $f'$ coincide, we can conclude that they intersect $\partial F$ the same number of times.
	As $f$ intersected $\partial F$ at most once, the same holds for $f'$ as well.
	
	We have seen that we can add an arbitrary edge.
	Repeating this process for all edges we get that the whole Delaunay-graph can be drawn in the plane as required.
\end{proof}

\section{Discussion}\label{sec:discussion}
\subsection{More general families}
We start by discussing the impossibility of several possible strengthenings of Theorem \ref{thm:planedelgraph}.

First, the theorem is trivially not true for a family of arbitrary convex regions as already the Delaunay-graph may not be planar (let alone aligned). Indeed, take $5$ points in general position to be $S$ and take $\F$ to be the family of the following convex regions: we connect any pair of points of $S$ with (slightly thickened) closed segments. The Delaunay-graph of $S$ with respect to $\F$ is $K_5$ which is non-planar.

Nevertheless, planarity of the Delaunay-graph was recently proved also in much more general settings \cite{K18,RR18}, in the latter also proving something stronger, the existence of a planar support (we discuss this in detail a bit later).  On the other hand, we now show that even for the weakest among these more general cases we already cannot guarantee that the drawing is aligned with the family. These generalizations consider $k$-admissible and more generally, non-piercing regions. A family of Jordan-regions $\F$ is \emph{non-piercing} if for every pair of regions $F,F'\in \F$ both $F\setminus F'$ and $F'\setminus F$ are connected and their boundaries intersect finite many times. A family of Jordan-regions is \emph{$k$-admissible} if it is non-piercing and for every pair of regions $F,F'\in \F$ their boundaries intersect at most $k$ times. Notice that the $2$-admissible families of regions are exactly the families of pseudo-disks. So even though even non-piercing families have planar Delaunay-graphs \cite{RR18}, it is easy to see that already for $4$-admissible families (which are also $k$-admissible for every $k\ge 4$ and also non-piercing; notice also that if touchings are not allowed, then a $3$-admissible family must also be $2$-admissible) the Delaunay-graph may be impossible to align with the family. Indeed, take two regions $F,F'$ whose boundaries intersect $4$ times and for which $F\cap F'$ is disconnected. Now if $S$ consists of two points, in different connected components of $F\cap F'$, then it is impossible to connect them inside $F\cap F'$ which would be required in an aligned drawing. See Figure \ref{fig:4admissible}.

\begin{figure}[t]
	\centering
	\includegraphics[height=3cm]{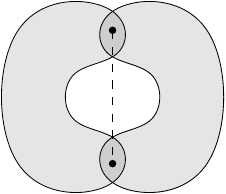}
	\caption{Delaunay-graph of a $4$-admissible family that cannot be aligned.}
	\label{fig:4admissible}	
\end{figure}
Thus, we conclude that pseudo-disk families are in this sense the most general natural families for which an aligned planar drawing of the Delaunay-graph always exists.

However, in a different fashion it is possible to say something positive about the drawing of the Delaunay-graph of a point set with respect to a $k$-admissible family for arbitrary $k$. Given a finite set of points $S$ and a finite $k$-admissible family of regions for some $k$, such that every region contains exactly two points from $S$ (i.e., all regions of $\F$ define a Delaunay-edge of the Delaunay-graph $D$ of $S$ with respect to $\F$), in \cite{PR} the following was proved. There exists a multigraph $D'$ which we get by multiplying edges of $D$ and a drawing of $D'$ in the plane without crossings such that for every region $F$ of $\F$ there exists an edge of $D'$ connecting the two points of $F\cap S$ and whose drawing lies completely in $F$. In other words, an aligned drawing is possible even for non-piercing families, provided that we only care about regions containing two points and that we allow to draw multiple copies of an edge.

\subsection{Infinite pseudo-disk families}
One can wonder whether the assumption that $\F$ is finite was necessary in Theorem \ref{thm:planedelgraph}. After all, in Theorem \ref{thm:classic} we did not make this assumption. The following example shows that for infinite families of pseudo-disks we cannot guarantee a planar aligned drawing of the Delaunay-graph (see Figure \ref{fig:infcounterex} for an illustration of the following construction): let $x=(0,0)$ and $y=(1,1)$ and $z=(1,-1)$ the three points of $S$, and let $p=(1,0)$ ($p$ is not in $S$). Our pseudo-disk family contains for all $k\ge 10$ two regions. First we take the union $f_{xy}$ of the two segments connecting $x$ to $p$ and $p$ to $y$, and take the Minkowski sum of $f_{xy}$ with a ball of radius $\frac 1{2k}$. We do the same for the union $f_{xz}$ of the two segments connecting $x$ to $p$ and $p$ to $z$ and a ball of radius $\frac 1{2k+1}$. It is easy to see that this is a pseudo-disk family. The Delaunay-graph of this family has two edges, $xy$ and $xz$ and it is easy to see that an aligned drawing of it is unique, $xy$ must be drawn as $f_{xy}$ and $xz$ must be drawn as $f_{xz}$. This is not a planar drawing.

\begin{figure}[t]
	\centering
	\includegraphics[height=3cm]{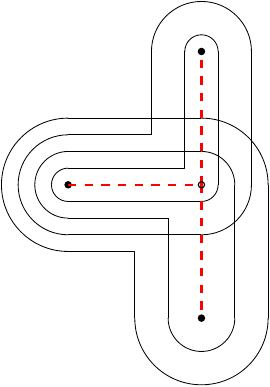}
	\caption{An inifinite family of pseudo-disks for which the Delaunay-graph does not have an aligned planar drawing.}
	\label{fig:infcounterex}	
\end{figure}

Nevertheless a natural relaxation of Theorem \ref{thm:planedelgraph} is true for infinite families. We claim that there exists a drawing of the Delaunay-graph of $S$ with respect to $S$ which is aligned with $\F$ and is a limit of planar drawings. This relaxation allows overlapping parts of the drawn edges as far as they do not cross. This definition is an extension of the definition of weakly simple drawings of curves (see the definition, e.g., in \cite{chang}) to multiple curves forming the edges of a graph.

Recall that we defined pseudo-disks to be closed and bounded. 
First, if $|\F|$ is uncountable, then for any $p,q\in S$ consider the pseudo-disk family $\F_{p,q}=\{F_{p,q}\in \F\mid p,q\in F_{p,q}\}.$
The intersection $\cap \F_{p,q}$ 
can also be obtained as the intersection of countably many members of $\F_{p,q}$ because the plane is a hereditary Lindel\"of space (i.e., every union of open sets has a countable subunion that is equal to it).
Therefore, it is sufficient to keep only countable many members of $\F_{p,q}$ to get the same containment requirement about the edge $pq$.
Thus from now on, we can suppose that $\F$ is countable.

Now if $|\F|$ is countable, then we start by a finite subfamily $\F_0$ of $\F$ which defines the same Delaunay-hypergraph as $\F$ (such a subfamily exists). Now we add the rest of the regions in $\F$ to $\F_0$ one by one, and each time we apply Theorem \ref{thm:planedelgraph}. This way we get an infinite series of planar drawings. Since every drawing is a closed subset of a compact set, there is a subseries of these drawings which has a limit. The limit graph is not necessarily planar but by definition it is the limit of planar graphs. Furthermore every pseudo-disk $F$ containing two points $p,q\in S$ gets at some point into the family, and from then on the drawing of the edge $pq$ is inside $F$, which implies that the limit of these drawings in also inside $F$ (as $F$ is closed), that is, the drawing of the final graph is aligned with $F$, as required.

\subsection{Planar supports}\label{sec:planarsupports}

We finish this section by showing the connections of our main result to planar supports. A \emph{planar support} of $S$ with respect to $\F$ is a planar graph $G$ on $S$ such that for every $F\in\F$ the point set $H_F=F\cap S$ induces a connected subgraph of $G$. Note that such a graph must contain the Delaunay-graph as a subgraph so in particular the existence of a planar support implies that the Delaunay-graph is planar.

First we need the following (restated) lemma from \cite{Pin14}.

\begin{lem}[Pinchasi \cite{Pin14}]\label{lem:shrink2}
	Given a finite family $\F$ of pseudo-disks, if a pseudo-disk $F\in \F$ contains exactly $k$ points of $S$, one of which is $p\in S$, then for every $2\le l\le k$ there exists a set $F'\subset F$ such that $p\in F'$ and $|F'\cap S|=l$, and $\F\cup\{F'\}$ is again a family of pseudo-disks.
\end{lem}

\begin{defi}
	We say that a pseudo-disk family $\F$ is {\em (abstractly) shrinkable} over $S$ if for every $F\in \F$, every $p\in F\cap S$ and every $2\le l\le k$ there exists a set $F'\in \F$ such that $p\in F'$ and $|F'\cap S|=l$.\footnote{The word `abstractly' in the definition emphasizes that $F'\subset F$ is now not required.}
\end{defi}

\begin{cor}\label{cor:extend}
	Given a finite pseudo-disk family $\F$ and a finite point set $S$, $\F$ can be extended to a finite pseudo-disk family shrinkable over $S$.
\end{cor} 

An example of a shrinkable pseudo-disk family is the collection of all disks in the plane over a finite $S$ that does not contain four points on a circle.
More generally, instead of disks, the family of all homothets of any convex set with a smooth boundary is shrinkable over a finite $S$ that does not contain four points on the boundary of a homothet. The family of all homothets of a convex polygon $C$ is also shrinkable over a finite $S$ if $S$ in general position with respect to $C$, that is, there is no homothet of $C$ that contains at least four points of $S$ on its boundary and no two points in $S$ define a line that is parallel to a line-segment which is a part of the boundary of $C$ \cite{tuple}.

\begin{lem}\label{lem:shrinkable}
	Given a pseudo-disk family $\F$ shrinkable over a finite point set $S$, for every $F\in\F$ the subgraph of $\D(S,\F)$ (the  Delaunay-graph of $S$ with respect to $\F$) induced by $F\cap S$ is a connected graph.
\end{lem}
\begin{proof}
	Using that $\F$ is shrinkable, applying the definition for an arbitrary point $p\in F\cap S$, there is an another point $q\in F\cap S$ such that there is an $F'\in \F$ for which $F'\cap S=F\cap S\setminus \{q\}$.
	By induction, the Delaunay-graph restricted to $F'\cap S$ is connected.
	Applying now the definition of shrinkability to $q$, the same holds for some other point $q'\in F\cap S$.
	Since $F\cap S\setminus \{q\}$ and $F\cap S\setminus \{q'\}$ both induce connected graphs, so does $F\cap S$ unless $F\cap S=\{q,q'\}$, but in this latter case $(q,q')$ is an edge of the Delaunay-graph because of $F$.
	This finishes the proof.
\end{proof}

Using that every pseudo-disk family can be extended to a shrinkable family by Corollary \ref{cor:extend}, Lemma \ref{lem:shrinkable} implies that given a finite point set $S$ and a finite pseudo-disk family $\F$, $S$ has a planar support with respect to $\F$. This has been shown earlier for pseudo-disks and more generally for $k$-admissible regions for every $k$ \cite{PR} (this was recently further generalized by \cite{RR18} where they show a much more general result about so-called intersection hypergraphs where not only the hyperedges but also the vertex set corresponds to a family of non-piercing regions). However, for pseudo-disks Theorem \ref{thm:planedelgraph} implies that a planar support which is aligned with $\F$ also exists (while we have seen that in the more general cases already the Delaunay-graph may not be possible to draw aligned with $\F$).

\begin{cor}\label{cor:planarsupport}
	Given a finite pseudo-disk family $\F$ and a finite point set $S$, $S$ has a planar support with respect to $\F$ in which every edge $pq$ lies inside every pseudo-disk containing both $p$ and $q$.
\end{cor}

\begin{proof}
	We extend $\F$ using Corollary \ref{cor:extend} to a family $\F'$ which is shrinkable over $S$. We draw the Delaunay-graph of $\F'$ using Theorem \ref{thm:planedelgraph} which is a planar support by Lemma \ref{lem:shrinkable} and furthermore every edge $pq$ lies inside every pseudo-disk (of $\F'$ and thus also of $\F$) containing both $p$ and $q$ due to Theorem \ref{thm:planedelgraph}, as required.
\end{proof}

\subsubsection*{Acknowledgment}

We thank all reviewers for the careful reading of the manuscript, and for calling our attention to related literature (\cite{KU13}), and Eyal Ackerman for several useful comments which substantially improved the manuscript.

\end{document}